\documentclass[a4paper]{amsart}

\usepackage[utf8]{inputenc}
\usepackage{xcolor}
\usepackage{amsmath}
\usepackage{amssymb}
\usepackage{amsfonts}
\usepackage{verbatim}
\usepackage{amsthm}
\usepackage{geometry}
\usepackage{hyperref} 
\usepackage{cleveref}
\usepackage{csquotes}

\PassOptionsToPackage{unicode}{hyperref}
\PassOptionsToPackage{naturalnames}{hyperref}

\usepackage{tikz}
\usetikzlibrary{shapes,positioning}
\usetikzlibrary{calc}
\usetikzlibrary{matrix,arrows,patterns.meta}
\usepackage{graphicx}
\tikzset{c/.style={every coordinate/.try}}
\newtheorem{thm}{Theorem}[section]
\newtheorem{proposition}[thm]{Proposition}
\newtheorem{coro}[thm]{Corollary}
\newtheorem{lemma}[thm]{Lemma}
\newtheorem{expl}[thm]{Example}
\theoremstyle{definition}
\newtheorem{definition}[thm]{Definition}

\theoremstyle{remark}
\newtheorem{remark}[thm]{Remark}

\newcommand{\calP}{\mathcal{P}}
\newcommand{\calH}{\mathcal{H}}
\newcommand{\calL}{\mathcal{L}}
\newcommand{\calC}{\mathcal{C}}
\newcommand{\calD}{\mathcal{D}}
\newcommand{\calO}{\mathcal{O}}
\newcommand{\calR}{\mathcal{R}}
\newcommand{\calT}{\mathcal{T}}
\newcommand{\calX}{\mathcal{X}}
\newcommand{\fqm}{\mathbb{F}_{q^m}}
\newcommand{\fq}{\mathbb{F}_{q}}
\newcommand{\fqo}{\mathbb{F}_{q_0^2}}
\newcommand{\F}{\mathbb{F}}

\newcommand{\PP}{\mathbb{P}}

\newcommand{\Tr}[1]{\operatorname{Tr}_{\mathbb{F}_{q^m}/\fq}\left(#1\right)}
\newcommand{\set}[1]{\left\{#1\right\}}

\newcommand{\Span}[1]{\operatorname{Span}\left(#1\right)}
\newcommand{\LT}[1]{\operatorname{LT}\left(#1\right)}
\newcommand{\Supp}{\operatorname{Supp}}
\newcommand{\Div}{\operatorname{Div}}

\newcommand{\GRS}{\operatorname{\mathsf{GRS}}}

\newcommand{\degab}[1]{\deg_{a,b}\left(#1\right)}

\title{Goppa--like AG codes from $C_{a,b}$ curves and their behaviour under squaring their dual}

\author{Sabira El Khalfaoui}
\address{Univ Rennes, IRMAR - UMR 6625, F-35000 Rennes, France}
\email{sabira.elkhalfaoui@univ-rennes1.fr}

\author{Mathieu Lhotel}
\address{Laboratoire de Mathématiques de Besançon, UMR 6623 CNRS Université de Bourgogne Franche-Comté, France}
\email{mathieu.lhotel@univ-fcomte.fr}

\author{Jade Nardi}
\address{Univ Rennes, CNRS, IRMAR - UMR 6625, F-35000 Rennes, France}
\email{jade.nardi@univ-rennes1.fr}

\date{\today}

\begin{document}
\maketitle

\begin{abstract}
In this paper, we introduce a family of codes that can be used in a McEliece cryptosystem, called \emph{Goppa--like AG codes}. These codes generalize classical Goppa codes and can be constructed from any curve of genus $\mathfrak{g} \geq 0$. Focusing on codes from $C_{a,b}$ curves, we study the behaviour of the dimension of the square of their dual to determine their resistance to distinguisher attacks similar to the one for alternant and Goppa codes developed by Mora and Tillich \cite{MT21}. We also propose numerical experiments to measure how sharp is our bound.
\end{abstract}

\section*{Introduction}

\subsection*{McEliece crytosystem}
McEliece cryptosystem is one of the last code--based candidates for standardization of post--quantum cryptrographic to the NIST competition since the third round. It guarantees the smallest ciphertexts among all the candidates, but it suffers from the largest public keys. Over the past forty years, there were many attempts in replacing the family of binary Goppa codes by other structured families of codes in order to reduce the key size.

The security of McEliece cryptosystem is based on two assumptions: $(i)$ in average, it is hard to decode $t$ errors for codes having the same parameters as the codes used as keys and $(ii)$ it is difficult to distinguish codes used as public keys from arbitrary ones. When proposing another family of codes, one must ensure that both of the hardness assumptions are still valid.
Algebraic geometry (AG) codes appear to be good candidates for McEliece cryptosystem, since they are evaluation codes built from algebraic curves (GRS codes are a typical example of AG codes on a genus zero curve). Moreover, they also come with an efficient decoding algorithm (see the survey of H\o hold and Pellikaan \cite{HP95}).
In 1996, Janwa and Moreno \cite{JM96} proposed to use AG codes, concatenation or subfield subcode of these codes in McEliece cryptosystem. As for concatenated ones, they were broken by Sendrier \cite{Sen94}. For AG codes, Faure and Minder proposed an attack in \cite{FM08,Min07,Fau09} for codes of genus at most $2$. The scheme based on AG codes have been completely broken by Couvreur, Marquez-Corbella and Pellikaan \cite{CMR17}, who proposed a filtration--based attack on AG codes for any genus, enabling decoding just by handling the public key and without knowledge of the curve and/or the divisor. However, the authors underlined that subfield subcodes of AG codes (SSAG codes in short) are resistant to this filtration attack. Moreover, some of these codes have a good designed minimum distance, such as Cartier codes \cite{Cou14}. This resistance to structural attacks, as well as their good parameters motivated several works on SSAG codes. Barelli \cite{B18} studied the security of quasi--cyclic SSAG codes designed from cyclic covers of the projective line and of plane curves. Recently, Zhao and Chen \cite{ZC22} analysed the parameters and the decoding performance of one--point elliptic subfield subcodes, showing that in the binary case, decoding results on these codes outperform the similar rate in the case of BCH codes. Also, the authors in \cite{PJ14,EKN21} focus on one--point Hermitian codes, and manage to compute the exact dimension of their subfield subcodes, in order to get a better estimate for the key size when using them in McEliece cryptosystem.

\subsection*{Subfield subcodes of AG codes}
Let $\fqm$ be a finite extension of $\fq$. Given an AG code $\calC :=C_{\calL}(\calX,\calP,G) \subseteq \fqm^n$, its subfield subcode over $\fq$ is defined by 
$$C_{\calL}(\calX,\calP,G)|_{\fq} := \calC \cap \fq^n.$$  
In this paper, we introduce a specific class of SSAG codes, namely Goppa--like AG codes. The idea of this construction is to introduce  some randomness into the family of SSAG codes drawn as private key by mimicking the role of the Goppa polynomial in the case of classical Goppa codes. Given an effective divisor $D\in \Div(\calX)$ on the curve, we will consider SSAG codes associated to divisors of the form $D+(g)$, for a large family of functions $g \in \fqm(\calX)$ on the curve. A Goppa--like AG code is then defined as the subfield subcode over $\fq$ of $C_{\calL}(\calX,P,D+(g))^{\perp}$. In the genus zero case, this construction gives codes which are Hamming equivalent to classical Goppa codes. With a good choice of the curve and divisor, it is possible to encode and decode these codes in a timely manner; this generalization has the potential to significantly improve on the original McEliece proposal. This encourages the present study of structural attacks on Goppa--like AG codes. 

\subsection*{Distinguisher attack}
Our starting point is the paper by Mora and Tillich \cite{MT21}, which benefited from the trace structure of the dual code of a subfield subcode to display a distinguisher for high rate Goppa codes. Their techniques rely on two features of GRS codes. First, their behaviour with respect to the Schur product is well known. From this, the authors gave a sharp upper bound for the dimension of the square of the dual of alternant codes. In the particular case of Goppa codes, where the multiplicator has the form $\mathbf{y}=(g(x_1)^{-1},g(x_2)^{-1},\dots,g(x_n)^{-1})$ for some degree $r$ polynomial $g$, they managed to get an even sharper upper bound by performing Euclidean division by powers of $g$.
As Goppa--like AG codes extend Goppa codes, it is natural to wonder if one can derive a structural attack on these codes from Mora and Tillich's attack. The genus of the curve plays a significant role in the parameters of AG codes. The greater the genus, the more $\F_q$--points the curve $\calX$ may have and so the longer the code can be. But the parameters $[n,k,d]$ of an AG code satisfy $n-\mathfrak{g}+1 \leq k+d \leq n+1$, which means that an AG code is $\mathfrak{g}$--far from optimality. Also, the correction capacity naturally suffers from a big genus: the naive correction algorithm can correct up to $\frac{d-1-\mathfrak{g}}{2}$ errors. Only refined techniques, based on error locating arrays \cite{CP20}, manage to remove the term related to the genus. Therefore, caution should be exercised when it comes to the impact of the genus on the properties of AG codes.
Even if Mora and Tillich's attack does not threaten the Goppa codes used in the NIST competition, SSAG--based propositions may be vulnerable to a similar structural attack. 

\medskip

\subsection*{Contributions of this paper}
In the present work, we are interested in the case of Goppa--like one--point AG codes from $\calC_{a,b}$ curves, a class of curves introduced in \cite{Miu93}. These curves have been extensively studied and they are especially interesting to design AG codes. As we know explicit monomial basis of the AG code associated to any multiple of its unique point at infinity, it allows us to efficiently encode one--point codes \cite{BRS21}. Furthermore, they remain quite general: for examples, elliptic curves, Kummer or Artin-Schreier curves are $C_{a,b}$ curves. It is also natural to wonder how the genus affects the distinguisher: in particular, we give a sufficient condition of the minimal degree our divisor has to satisfy in order to mount the distinguisher. This bound is increasing with the genus of the curve, and coincides with the one given in \cite{MT21} in the case of classical Goppa codes. Consequently, when the genus gets higher, we are not able to distinguish codes associated to low degree divisors. In the worst cases, we might not be able to distinguish anything.
AG codes, as generalizations of GRS codes, have exactly the same behaviour with respect to the Schur product. Moreover, some well--chosen AG codes are defined by the evaluation of \textit{multivariate} polynomials. In this case, we prove that performing division algorithms via Gr\"obner bases enable us to estimate the dimension of the square of the dual of Goppa--like AG codes. Even better, computations tend to show that the bound we obtain on the dimension is sharp whenever the one--point Goppa--like code seems random (\emph{i.e} the function $g$ is sufficiently random). The counterpart is, as it was the case in \cite{MT21} for classical Goppa codes, that we can only distinguish high rate codes. More precisely, our maximum distinguishable rate is roughly the same as in \cite{MT21}. As comparison of their results and ours is carried out in the case of Goppa--like codes defined over an elliptic curve. However, as previously discussed, if the genus becomes too large, the distinguisher may become ineffective: in particular, we show that Goppa--like AG codes constructed from the Hermitian curve resist our distinguisher, which is encouraging if we intend to base a McEliece's like cryptosystem on such class of codes.
However, as it is already the case for the distinguisher of classical Goppa codes \cite{MT21}, it seems complicated to turn this distinguisher into an efficient structural attack. But, having an algebraic explanation of the structure of the square of the dual of one--point Goppa--like codes is still desirable if we want to perform an attack using square code considerations. It also provides a rigorous method for selecting initial parameters that ensure the security of our cryptosystem.

\subsection*{Application of Goppa-like AG codes to McEliece cryptosystem}

For a McEliece cryptosystem constructed on a $[n,k]_q$ linear code $\calC$ with an error capability $t$ , the public key has size $k(n-k)\lceil \log_2(q) \rceil$, whenever the generator matrix $\mathbf{G}$ of $\calC$ is in systematic form. The choice of such a code should not be limited to the binary case: Peters \cite{petersIsd} showed that for $q > 2$, the subfield subcode construction over $\fq$ can improve the key size while maintaining the same level of security against a decoding attack. 

\emph{Information--set decoding} (ISD) algorithm, introduced by Prange \cite{prange}, is the best known algorithm for decoding a random--looking code without any knowledge of its structure. Given a $(n-k) \times n$ parity check matrix $\mathbf{H}$ and a column vector $\mathbf{s}$ of length $n-k$ obtained through $\mathbf{H}$, the total complexity of finding the row vector $\mathbf{e} \in \fq^n$ of weight $t$ and satisfying $\mathbf{s}=\mathbf{H}\mathbf{e}^{T}$ is given by 
$$ C_\text{Prange}= \frac{\binom{n}{t}}{\binom{n-k}{t}}C_\text{Gauss}(n,k,q),$$

where $C_\text{Gauss}(n,k,q)$ is the cost of the Gauss--Jordan elimination of an $k\times n$ matrix over $\mathbb{F}_q$. We analyse the parameters of Goppa--like AG codes based on information set decoding attack.
 
In this work, we also investigate Goppa--like AG codes over finite fields with different characteristics by analysing the computational cost of the ISD algorithm for parameters relevant to post--quantum cryptography. In Table \ref{table:goppa-herm}, we provide parameters that resist the distinguisher given in this paper and improve key sizes compared to the subfield subcodes of 1--point Hermitian codes parameters reported in \cite[Tables~2 and 3]{EKN21crypto}, which already reduced key sizes compared to binary Goppa codes.

\begin{table}[h]
	\begin{center}
		\begin{tabular}{|c|c|c|c|c|c|c|}
			\hline
			$q$&  $s$ & $n$ & $k$ & $t$ & Prange complexity & Key--Size(bit)\\
			\hline \hline
			
		${11}$	&265&1320& 898& 77& 153& 1\,136\,868 \\
		
			\hline \hline
			${13}$&312&2188& 1718& 77& 198& 2\,422\,380  \\
			
			\hline 
		${16}$&354& 4078& 3608& 56& 199& 6\,783\,040   \\
		
			\hline \hline
		${13}$& 490& 2189& 1363& 166& 270& 3\,377\,514 \\
		
			\hline 
			${16}$& 460 &4080& 3398& 109& 313& 9\,269\,744\\
			\hline
		\end{tabular}
		\vspace*{0.3em}
		\caption{Goppa-Like Hermitian codes parameters $\Gamma(\calP,sP_\infty,g)$ over $\F_{q^2}$.} \label{table:goppa-herm}
	\end{center}
\end{table}


\subsection*{Outline of the paper}

The paper is organized as follows. Section \ref{sec:preli} is dedicated to basic definitions about linear codes, subfield subcodes and trace codes, AG codes and $C_{a,b}$ curves. In Section \ref{sec:Goppa}, we introduce the notion of Goppa--like AG codes and we give a first bound for the dimension of the square of their dual. In Section \ref{sec:AG-C_a,b}, we refine the bound for Goppa--like AG codes from $C_{a,b}$ curves, associated to one--point divisors. Section \ref{sec:analysis} concludes this paper with a discussion about this bound on elliptic curves and on the Hermitian curve.

\section{Preliminaries}\label{sec:preli}
\subsection{Linear codes, subfield subcodes and trace codes}

In this section, we briefly introduce some notation and basic definitions for linear codes, subfield subcodes, and trace codes. Furthermore, we present significant results that employ component--wise product and trace map. For further details about linear codes, we refer the reader to \cite{MS86}.

\noindent Let $\fqm$ be a finite extension of the field $\fq$ with $q$ elements, where $m \geq 1$. A \emph{linear code} $\calC$ over $\fqm$ is a vector subspace of $\fqm^n$. The integer $n$ is called its \emph{length} and we denote by $k$ its \emph{dimension} and by $d$ its \emph{minimum distance}. We say that $\calC$ is a $[n,k,d]_{q^m}$ code or that it has parameters $[n,k,d]_{q^m}$ (we may omit the minimum distance in this notation). Given a linear code of length $n$, its dual code $\calC^{\perp}$ is defined by 
\[\calC^{\perp}=\left\lbrace \mathbf{x} \in \fqm^n \mid \mathbf{c} \cdot \mathbf{x}=0, \text{ for all } \mathbf{c} \in \calC \right\rbrace,\]  
where $\cdot$ denotes the usual scalar product. It is easy to verify that if $\calC$ is a $[n,k]_{q^m}$ code, then $\calC^{\perp}$ is a $[n,n-k]_{q^m}$ code.
A generator matrix $\mathbf{G}$ of a $[n,k]_{q^m}$ linear code $\calC$ is a $k \times n$ matrix over $\fqm$ whose rows are a basis of $\calC$ as a vector space, while a parity check matrix $\mathbf{H}$ of $\calC$ is a any generator matrix of its dual.
For a linear code with length $n$ and dimension $k$, its rate is defined by the ratio $R := \frac{k}{n}$.
The Schur product of two vectors $\mathbf{a}$,$\mathbf{b} \in \fqm^n$ is defined as 
\[ \mathbf{a} \star \mathbf{b} := (a_1b_1,\cdots,a_nb_n). \]
It can be extended to codes in the following way: If $\calC$ and $\calD$ are two codes over $\fqm$ with same length $n$, their Schur product is defined as the following component--wise product:
\[ \calC \star \calD := \Span{\mathbf{c} \star \mathbf{d} \mid \mathbf{c} \in \calC, \mathbf{d} \in \calD}. \]
Moreover, if $\calC = \calD$, we call $\calC^{\star 2} := \calC \star \calC$ the square of $\calC$. The following lemma gives a first estimation of the dimension of Schur product of codes.

\begin{lemma}[{\cite[Proposition~10]{MT21}}] \label{lem:known_bounds}
Let $\calC$ and $\calD$ be two linear codes over $\fqm$ with same length $n$ and respective dimension $k_{\calC}$ and $k_{\calD}$. We have
\begin{enumerate}
 \item $\dim_{\fqm}(\calC \star \calD) \leq \min\{k_{\calC}k_{\calD},n\}$;
 \item\label{it:dim_C^2} $\displaystyle \dim_{\mathbb{F}_{q^m}}(\calC^{\star2}) \leq \mathrm{min}\left(n,\binom{k_{\calC}+1}{2}\right)$.
 \end{enumerate}
\end{lemma}

In particular, by \cite[Theorem~2.3]{CCMZ15}, if $\calC$ is sufficiently random and $\calC^{\star2}$ does not fill the full space, we expect to have 
\[ \dim_{\fqm}(\calC^{\star2}) = \binom{k_{\calC}+1}{2}.\]

Given a linear code $\calC$ over $\fqm$, there exist two constructions of codes on the subfield $\fq$, namely its subfield subcode and its trace code. The \emph{subfield subcode} of $\calC$, denoted by $\calC|_{\fq}$, is the linear code over $\fq$ defined by 
\[\calC|_{\fq}=\calC \cap \mathbb{F}_q^n.\]
Again, if $\calC$ is a $[n,k,d]_{q^m}$ code, then $\calC|_{\fq}$ is a $[n,\geq n-m(n-k),\geq d]_q$ code.
Let $\operatorname{Tr}_{\mathbb{F}_{q^m}/\fq}$ be the trace operator on $\mathbb{F}_{q^m}$ with respect to $\mathbb{F}_q$, \emph{i.e.} defined by
\[\Tr{x} = x + x^q + ... + x^{q^{m-1}},\]
for any $x \in \fqm$. We can extend this definition to any vector $\mathbf{c} \in \fqm^n$ by $$\Tr{\mathbf{c}}= (\Tr{c_1},\cdots,\Tr{c_n}).$$ 
Given a linear code $\calC$ of length $n$ and dimension $k$ over $\fqm$, its \emph{trace code} over $\fq$ is the image under the trace operator, that is $\Tr{\calC}$. It is a linear code of length $n$ over $\fq$, whose dimension satisfies
\begin{equation}\label{eq:dim_trace}
\dim_{\mathbb{F}_q} \Tr{\calC} \leq \min\{mk,n\}.
\end{equation}
Subfield subcodes and trace codes are related by the duality operator, as stated by Delsarte's theorem.

\begin{thm}[{\cite[Delsarte's theorem]{Del75}}] \label{th:delsarte}
Let $\calC$ be a linear code of lenght $n$ over $\fqm$. Then
\[\left(\calC \cap \fq^n\right)^{\perp} = \Tr{\calC^{\perp}}.\]
\end{thm}

\subsection{AG and SSAG codes} \label{section:AG_codes}

\subsubsection{Definitions}

Let $\calX$ be a smooth and irreducible projective curve over $\fqm$ of genus $\mathfrak{g}$. A \emph{divisor} on $\calX$ over $\fqm$ is a formal sum of places over $\fqm$, \emph{i.e.} of the form $G=\sum \nu_P(G) P$, where $\nu_P(G)$ are integers which are all zero except for a finite number of places $P$. We denote by $\Div(\calX)$ the set of $\fqm$--divisors on $\calX$ (we omit the dependence on $\fqm$).
Given $G \in \Div(\calX)$, we define its \emph{support} $\Supp(G)$ as the finite set of places $P$ such that $\nu_P(G)$ is non--zero and its \emph{degree} as $\deg G=\sum \nu_P(G) \deg(P)$. We say that a divisor $G \in \Div(\calX)$ is \emph{effective} if for all $P \in \Supp(G)$, we have $\nu_P(G) \geq 0$, in which case we write $G \geq 0$. This permits to define an order on the group of divisors by setting $G_1 \geq G_2$ if and only if $G_1-G_2 \geq 0$.

A non--zero function $f \in \fqm(\calX)$ defines a divisor, called \emph{principal}, denoted by $(f)=\sum \nu_P(f) P$. The \emph{Riemann--Roch space} of a divisor $G$ is defined as the $\fqm$ vector space
$$ \calL(G) := \set{f \in \fqm(\calX) \mid (f) + G \geq 0} \cup \set{0},$$
of dimension $\ell(G)$.
Let $\calP \subseteq \fqm(\calX)$ be a set of $n$ distinct rational points such that $\Supp(G) \cap \calP = \varnothing$.
We can consider the AG code 
$$\calC := \calC_{\calL}(\calX,\calP,G) := \set{\left(f(P)\right)_{P \in \calP} \mid f \in \calL(G)},$$
which is a $[n, k \geq \ell(G)]_{\fqm}$ code. If $n > \deg G$, then its dimension is exactly equal to $\ell(G)$ and its minimum distance is bounded from below by the \emph{designed distance} $d^*=n-\deg G$.\\
From the Riemann--Roch theorem, we have
$$ \ell(G) \geq \deg(G) +1 - \mathfrak{g},$$ 
with equality if $\deg(G) \geq 2\mathfrak{g}-1$.\\

The dual of an AG code can be described as a residue code (see \cite{Sti09} for more details), \emph{i.e.}
$$ C_{\calL}(\calX,\calP,G)^{\perp} = C_{\Omega}(\calX,\calP,G).$$ 
Residue and evaluation codes are connected by the following result.

\begin{proposition} [{\cite[Proposition~2.2.10]{Sti09}}] \label{prop:dual_AG_codes}
Let $C_{\calL}(\calX,\calP,G)$ be an AG-code defined on a curve $\calX$. Then 
\[C_{\Omega}(\calX,\calP,G) = C_{\calL}(\calX,\calP,G^{\perp}),\]
with $G^{\perp} = D_{\calP}-G+W$, where $D_{\calP} = \sum\limits_{P \in \calP} P$ and $W=(\omega)$ is a canonical divisor such that for all $P \in \calP$, $\nu_P(\omega)=-1$ and $\mathrm{Res}_{\omega}(P)=1$. 
\end{proposition}


In general, it is hard to find the true dimension of the subfield subcode of an AG code, but a trivial estimate can be derived from Delsarte's theorem (Theorem~\ref{th:delsarte}):
\begin{equation}\label{eq:dim_ssag}
	\dim_{\fq} C_{\calL}(\calX,\calP,G)|_{\fq} \geq n - m\dim_{\fqm} C_{\Omega}(\calX,\calP,G).
\end{equation}

The minimum distance of the SSAG code $C_{\calL}(\calX,\calP,G)|_{\fq}$ is at least the one of the AG code $C_{\calL}(\calX,\calP,G)$ above. It is thus bounded from below by the designed distance $d^*=n-\deg G$.%

The structure of the AG codes may provide sharper bounds on the dimension of subfield subcodes and trace codes of AG codes.

\begin{thm}[{\cite[Theorem~9.1.6]{Sti09}}]\label{thm:dim_ssag}
	With the notation as above, let $G_1$ be any divisor such that 
	\begin{equation} \label{eq:divisor_G_1}
		G \geq qG_1 \ \mathrm{and} \ G \geq G_1.
	\end{equation}

	Then
\[	\dim_{\fq} \Tr{C_{\calL}(\calX,\calP,G)} \leq  \left\{\begin{array}{ll}
	m\left(\ell(G) - \ell(G_1)\right)+1 & \text{if } G_1 \geq 0, \\
	m\left(\ell(G) - \ell(G_1)\right) & \text{if } G_1 \not\geq 0, 
\end{array} \right.\]	
	and	
\[	\dim_{\fq} C_{\Omega}(\calX,\calP,G)|_{\fq} \geq  \left\{\begin{array}{ll}
	n-1-m\left(\ell(G) - \ell(G_1)\right) & \text{if } G_1 \geq 0, \\
	n-m\left(\ell(G) - \ell(G_1)\right) & \text{if } G_1 \not\geq 0. 
\end{array} \right.\]
\end{thm}

\begin{remark}\label{rk:G/q}
	The biggest divisor (with respect to the degree) satisfying the conditions in Equation \eqref{eq:divisor_G_1} is given by 
	\begin{equation}\label{eq:G/q}
		\left[ \frac{G}{q} \right]:= \sum\limits_{P \in \Supp(G^+)} \left\lfloor\frac{\nu_P(G^+)}{q}\right\rfloor P + \sum\limits_{P \in \Supp(G^-)}\nu_P(G^-)P,
	\end{equation}
	where $G^+$ and $G^-$ are effective divisors such that $G=G^+-G^-$.
\end{remark}

With additional hypotheses on $G$ and $\displaystyle \left[ \frac{G}{q} \right]$, {\cite[Theorem~1]{Le16}} gives an exact formula for such a code.

Regarding the parameters of subfield subcodes of differential codes, Wirtz \cite{W88} improved the bound on their the minimum distance.

\begin{thm}[{\cite[Theorem~2]{W88}}]\label{thm:Wirtz}
	Take the same notation as in Theorem \ref{thm:dim_ssag} and assume $\deg G_1 > 2 \mathfrak{g}-2$.
	Set $U:=\{P \in \Supp(G) \mid \nu_P(G) \geq 0 \text{ and } \nu_P(G) = q-1 \mod q\}$ and $G_U=\sum_{P \in U} P$. Then
	\[	\dim_{\fq} C_{\Omega}(\calX,\calP,G)|_{\fq} = 	\dim_{\fq} C_{\Omega}(\calX,\calP,G+G_U)|_{\fq}, \]
	hence the minimum distance of $C_{\Omega}(\calX,\calP,G)|_{\fq}$ satisfies
	\[d\left(C_{\Omega}(\calX,\calP,G)|_{\fq} \right) \geq \deg G + \deg G_U -2\mathfrak{g}+2.\]
\end{thm}

\subsubsection{First estimation of the dimension of the square of the trace of an AG code}

In this paper, we aim to bound the dimension of the square of the dual of a SSAG code. From Delsarte's theorem, this is equivalent to study the square of the trace of the corresponding AG code. 
This is possible thanks to the following result from \cite{MT21}, which is valid for any linear code.

\begin{proposition}[{\cite[Proposition~15]{MT21}}] \label{prop:Tr_BoundSchurSquare}
 Let $\calC$ be a linear code over $\fqm$. Then we have 
 \begin{equation} \label{eq:key_equation} \Tr{\calC}^{\star2} := ((\calC^{\perp}|_{_{{\mathbb{F}_q}}})^{\perp})^{\star2} \subseteq \sum\limits_{i=0}^{\lfloor{m/2} \rfloor} \Tr{\calC\star \calC^{q^i}},
 \end{equation}
Moreover, if $m$ is even, 
 \begin{equation} \label{eq:dim_m/2} \dim_{\fq} \Tr{\calC\star \calC^{q^{\frac{m}{2}}}} \leq \frac{m}{2} (\dim_{\fqm}\calC)^2.
\end{equation}
\end{proposition}
From this result, we can deduce an upper bound for the dimension of the square of the dual of any linear code.

\begin{coro} [{\cite[Corollary~16]{MT21}}]\label{coro:first_bound_square_of_trace}
 Let $\calC$ be any $\fqm$-linear code. Then 
 \begin{equation} \label{eq:mumford_bound}
  \dim_{\fq}\Tr{\calC}^{\star2} \leq m \cdot \dim_{\fqm}(\calC^{\star 2}) + \binom{m}{2} (\dim_{\fqm}(\calC))^2.
 \end{equation}
 Furthermore, if $\dim_{\fq} \Tr{\calC} = m \cdot \dim_{\fqm}(\calC)$, then 
 \[\dim_{\fq} \Tr{\calC}^{\star2} - \binom{\dim_{\fq} \Tr{\calC}+1}{2} \leq m \left( \dim_{\fqm} \calC^{\star 2} - \binom{\dim_{\fqm} (\calC)+1}{2}\right).\]

\end{coro}
The above corollary implies that if the dimension of a square code is smaller than that of a random code, namely
\[ \dim_{\fqm} (\calC^{\star 2}) < \binom{\dim_{\fqm} (\calC)+1}{2},\]
then this property is retained for the trace code, \emph{i.e.}
\[\dim_{\fq} \Tr{\calC}^{\star 2} < \binom{\dim_{\fq} \Tr{\calC}+1}{2}.\]
This is especially true for Reed--Solomon codes (see \cite{MT21}, Proposition 11) and more generally for AG codes.

\begin{proposition} [{\cite[Theorem~6]{Mum70}}] \label{prop:mumford_result}
 Let $F,G$ be two divisors on $\calX$ such that $\deg(G) \geq 2\mathfrak{g}+1$ and $\deg(F) \geq 2\mathfrak{g}$. Then
 \[ \calL(F) \cdot \calL(G) = \calL(F+G),\]
 where $\calL(F) \cdot \calL(G) := \Span{ f \cdot g : (f,g) \in \calL(F) \times \calL(G)}$.
\end{proposition}
As a consequence, for an AG code  $C_{\calL}(\calX,\mathcal{P},G)$, we have
\[ C_{\calL}(\calX,\mathcal{P},G)^{\star2} \subseteq C_{\calL}(\calX,\calP,2G).\]
with equality if $\deg(G) \geq 2\mathfrak{g}+1$.
If $\deg G \geq \mathfrak{g}$, applying the Riemann--Roch theorem to the divisors $G$ and $2G$ thus gives
\begin{equation}\label{eq:dim_square}
	\dim_{\fqm}(C_{\calL}(\calX,\mathcal{P},G)^{\star2}) \leq \dim_{\fqm} C_{\calL}(\calX,\calP,2G) = 2\deg(G)+1-\mathfrak{g} \leq \deg(G) + \dim_{\fqm}(C_{\calL}(\calX,\mathcal{P},G)),
\end{equation}
which is much smaller than the expected dimension given in Lemma \ref{lem:known_bounds} (\ref{it:dim_C^2}). (thus providing a distinguisher for AG codes). Combined with Equation \eqref{eq:key_equation}, we get a first upper bound on the dimension of the square of the trace of an AG code.

\begin{coro} \label{coro:1st_bound_mumford}
 Let $\mathcal{C} := C_{\calL}(\calX,\mathcal{P},G)$ be a $k$--dimensional AG code on $\calX$ associated with a degree $s \geq \mathfrak{g}$ divisor. Then
 \[ \dim_{\fq}\Tr{\calC}^{\star2} := \dim_{\fq} \left((\calC^\perp|_{\fq})^{\perp}\right)^{\star2}  \leq \binom{mk+1}{2} - \dfrac{m}{2} (k(k-1)-2s).\]
\end{coro}

\begin{proof}
 From Proposition \ref{prop:mumford_result} and Equation \eqref{eq:dim_square}, we have $\dim_{\fqm}(\calC)^{\star2} \leq 2s+1-g \leq k+s$. Thus, Equation \eqref{eq:mumford_bound} leads to
 \begin{align*}
  \dim_{\fq}\Tr{\calC}^{\star2} &\leq m(k+s) + \binom{m}{2}k^2 \\
  &\leq (2k+2s+mk^2-k^2) \dfrac{m}{2} \\
  &\leq (k(mk+1)-k^2+k+2s) \dfrac{m}{2} \\
  &\leq \binom{mk+1}{2} - \dfrac{m}{2}(k(k-1)-2s) .
 \end{align*}
\end{proof}

According to the above corollary, the dimension of the square of the dual of a SSAG code is less than the expected value for random codes (which is $\binom{mk+1}{2}$), due to the algebraic structure of AG codes. However, this bound does not fully benefit from this rich structure, notably the following property.

\begin{lemma}\label{lem:Schur-Product-Power}
	 Let $\mathcal{C} := C_{\calL}(\calX,\mathcal{P},G)$ be a $k$--dimensional AG code on $\calX$. For every $i \geq 0$, we have
	 \[\calC \star \calC^{q^i} \subseteq C_{\calL}(\calX,\mathcal{P},(q^i+1)G)\]
\end{lemma}

\begin{proof}
	Fix $i \geq 0$ and let $f_1,f_2 \in \calL(G)$. Then the product $f_1f_2^{q^i}$ belong to $\calL((q^i+1)G)$ as
	$$(f_1f_2^{q^i}) + (q^i+1)G = \left((f_1) + G\right) + q^i\left((f_2) + G\right) \geq 0.$$
	This proves the inclusion of spaces
	$$ \calL(G)\cdot \calL(G)^{q^i} \subseteq \calL((q^i+1)G),$$
	hence the inclusion of the associated codes.
\end{proof}

\begin{remark}
The property above for $i=0$ follows from Proposition \ref{prop:mumford_result}. To the best of our knowledge, there is no sufficient criterion for the equality for $i \geq 1$ in the literature. Given a basis $\set{f_1,\dots,f_k}$ of the Riemann--Roch space $\calL(G)$, the vector space $\calL(G)\cdot \calL(G)^{q^i}$ is spanned by the set $\set{f_u f_v^{q^i} \mid 1 \leq u \leq v \leq k}$. From our experiments, it may happen that the cardinality of this family is larger than $\ell((q^i+1)G)$ without the equality holding, which means that these generators may be linearly dependent in $\calL((q^i+1)G)$ and do not form a basis of $\calL(G)\cdot \calL(G)^{q^i}$.
\end{remark}

Thanks to Lemma \ref{lem:Schur-Product-Power}, it will be possible to better handle the terms $\Tr{\calC\star \calC^{q^i}}$ in Equation \eqref{eq:key_equation}. Our aim for the rest of the paper is to improve the bound of Corollary \ref {coro:1st_bound_mumford} in some specific cases.

\subsection{$C_{a,b}$ curves} \label{section:C_a,b_codes}
Throughout this paper, we will be dealing with algebraic geometry codes defined over a $C_{a,b}$ curve. This section is dedicated to some properties on this well--studied class of curves. For further details, we refer to \cite{Miu93}. 

\begin{definition} \label{def:C_ab_curves} 
Let $a,b$ be coprime positive integers. A $C_{a,b}$ curve over $\fqm$ is a curve $\calX_{a,b}$ having an irreducible, affine and non--singular plane model with equation
\begin{equation} \label{eq:equation_C_ab}
f_{a,b}(x,y) = \alpha_{0a}y^a + \alpha_{b0}x^b + \sum \alpha_{ij}x^iy^j = 0,
\end{equation}
where $f_{a,b} \in \fqm[X,Y]$ and the sum is taken over all couples $(i,j) \in \set{0,\cdots,b} \times \set{0,\cdots,a}$ such that $ai+bj < ab$.
\end{definition}
Any curve $\calX_{a,b}$ defined by an equation as in \eqref{eq:equation_C_ab} has a unique point at infinity, denoted by $P_{\infty}$. Moreover, as a plane curve, its genus is given by
\[\mathfrak{g}_{a,b}:=\dfrac{(a-1)(b-1)}{2}.\]

\subsubsection{The point at infinity $P_\infty$ and its regular ring $\calO_{P_\infty}$}

We will consider codes obtained by evaluating functions on $\calX_{a,b}$ which are regular everywhere, except maybe at the unique point at infinity $P_\infty$. These functions then belong to the ring 
\begin{equation}\label{eq:O_Pinf}
\mathcal{O}_{P_\infty}=\bigcup_{s \geq 0} \calL(s P_\infty)
\end{equation} %
where each Riemann--Roch space $\calL(s P_\infty)$ has an explicit basis as follows:
\begin{equation} \label{eq:basis_L(sP_inf)}
    \calL(sP_{\infty}) = \Span{x^iy^j \mid 0 \leq i, 0\leq j\leq a-1 \ \mathrm{and} \ ai+bj \leq s}.
\end{equation}

In summary, any function that is regular on all $\calX_{a,b}$ except maybe at $P_\infty$ is a bivariate polynomial in the functions $x$ and $y$. 

\begin{definition}[Weighted degree]
Given a monomial of the form $x^iy^j \in \calO_{P_\infty}$, we define its weighted degree by
\[ \degab{x^iy^j} := ai+bj.\]

From this degree, we can define a monomial order $\prec$ over $\calO_{P_\infty} \simeq \fqm[x,y]$ as follows: $x^uy^v \prec x^{u'}y^{v'}$ if
\begin{equation}\label{eq:def_mon_order}
\degab{x^uy^v} < \degab{x^{u'}y^{v'}} \text{ or } \left(\degab{x^uy^v} = \degab{x^{u'}y^{v'}}  \text{ and } u < u'\right).
\end{equation}

Any function $f \in \calO_{P_\infty}$ can be written in the form $$f = x^{\beta}y^{\alpha} + f'(x,y),$$
with $\alpha \leq a-1$ and $f' \in \calO_{P_\infty}$ such that any monomial $x^iy^j$ of $f'$ satisfies $ai+bj < \degab{x^{\beta}y^{\alpha}}$ and $j \leq a-1$. The leading term of $f$ with respect to the monomial order $\prec$ is thus defined by $\mathrm{LT}(f) := x^{\beta}y^{\alpha}$. This extends the definition of weighted degree to any such function by setting 
\[\degab{f} := \degab{\mathrm{LT}(f)}.\]
\end{definition}
It is easy to check that for any $f \in \calO_{P_\infty}$, its weighted degree $\degab{f}$ is equal to the smallest integer $s$ such that $f$ belongs to the Riemann--Roch space $\calL(sP_{\infty})$.
This way, any function in $\calL(sP_\infty)$ can be seen as a polynomial in $x$ and $y$ such that $\degab{f}\leq s$. 
\begin{remark}
It is worth noting that $\calO_{P_\infty}$ is a valuation ring, with valuation $v_{P_\infty}$. Then for every $f \in \calO_{P_\infty}$, we have $\degab{f}=-v_{P_\infty}(f)$. We prefer handling the weighted degree rather than the valuation due to its similarities with the degree of univariate 
polynomials. We will notably perform division, using Gr\"obner bases, and, as expected in the univariate case, the degree of the remainder is 
\textit{generally} smaller than the dividend's.
\end{remark}

\section{Goppa--like AG codes}\label{sec:Goppa}
\subsection{Definition, parameters and context in the literature}\label{subsec:def-Goppa}

Let $D$ be an effective divisor of positive degree $s$ on a smooth and irreducible projective curve $\calX$ over $\fqm$. Take a rational function $g \in \fqm(\calX)$ such that $g \notin \calL(D)$. Given a set of $\fqm$--points $\calP \in \calX(\fqm)$ such that $\calP \cap \Supp(g) = \varnothing$ and $\calP \cap \Supp(D) = \varnothing$, we consider the AG code
\[\calC := \calC_{\calL}(\calX,\calP,D+(g))=\set{\left(f(P)g(P)^{-1}\right)_{P \in \calP} \mid f \in \calL(D)}.\]
\begin{definition} \label{def:Goppa--like_AG_code}
The Goppa--like AG code associated to $\calC$ is defined as the subfield subcode of its dual code, \emph{i.e.}
\[ \Gamma(\calP,D,g) := \calC^{\perp}|_{\fq}.\]
\end{definition}%
Such a code has length $n = \# \calP$. As stated in \cite[Theorem~1]{JM96}, its dimension satisfies 
\[\dim_{\fq} \Gamma(\calP,D,g) \geq n-m \dim_{\fqm} \calC_{\calL}(\calX,\calP,D+(g)) =n-m(\deg D - \mathfrak{g} +1 )\]
if $2\mathfrak{g} - 2 < \deg D < n$. Its minimum distance is bounded from below by $\deg D - 2 \mathfrak{g} + 2$.

\begin{remark}
These estimations of the dimension and the minimum distance may be improved by Theorem~\ref{thm:dim_ssag} and \ref{thm:Wirtz}. Regarding the dimension, it is worth noting that, since $g \notin \calL(D)$, the divisor $G$ is not effective. Hence, any divisor $G_1$ satisfying the conditions of Equation \eqref{eq:divisor_G_1} is also non--effective, which means that
\[	\dim_{\fq}  \Gamma(\calP,D,g) \geq 	n-m\left(\ell(G) - \ell(G_1)\right). \]
Without additional conditions on the divisor $D$ and the function $g$, the divisor $G_U$ for $G=D+(g)$ defined in Theorem~\ref{thm:Wirtz} is zero. Generally, we cannot expect for a better bound for the minimum distance. 
\end{remark}

\subsubsection{Why the terminology \textit{Goppa--like}?}

In \cite{JM96}, Janwa and Moreno define \emph{Goppa codes} on smooth and irreducible projective curves. Compared to their definition, Definition~\ref{def:Goppa--like_AG_code} introduces a function $g$ which defines a multiplicator for the AG code over $\fqm$ that is algebraically related to the support $\calP$.

Introducing this function $g$ facilitates the use of SSAG as public keys for McEliece cryptosystem. Given an correcting error capability $t$, we can fix a divisor $D$ whose degree satisfies $\deg D \geq 2t + 2\mathfrak{g}+1$. Then the family of codes in which the public key is picked can be defined by running a family of functions $g$ outside $\calL(D)$.

\medskip

We prefer the terminology \emph{Goppa--like AG codes} instead of simply Goppa codes for two reasons. First, AG codes were historically called geometric Goppa codes. Our terminology involving both the words \textquote{Goppa} and \textquote{AG} removes this possible ambiguity. Second, we want to emphasize the use of a different curve than the projective line $\PP^1$, like we differentiate AG codes from Reed--Solomon codes.

This denomination is obviously motivated by the fact that here, the rational function $g$ plays the role of the Goppa polynomial. As described in \cite[Example~9.1.8]{Sti09}, Goppa codes are nothing but Goppa--like AG codes from the projective line $\calX=\PP^1$. In fact, given $r \geq 0$, the Generalized Reed--Solomon (GRS) code of degree $r$, support $\mathbf{x} \in \fqm^n$ and multiplier $\mathbf{y} \in (\fqm^*)^n$ is defined as
\[\GRS_r(\mathbf{x},\mathbf{y})=\{(y_1f(x_1),y_2f(x_2),\dots,y_nf(x_n)) \mid f \in \fqm[X] \text{ such that } \deg f < r \}.\]
Take a univariate polynomial $g$ of degree $r$ such that $g(x_i) \neq 0$ for every  $i \in \{1,\dots,n\}$. Then the Goppa code of order $r$ and support $\mathbf{x} \in \fqm^n$ is defined as
\[\Gamma_r(\mathbf{x},g)= \GRS_r(\mathbf{x},\mathbf{y})^\perp|_{\fq}\]
where $\mathbf{y}=(g(x_1)^{-1},g(x_2)^{-1},\dots,g(x_n)^{-1})$.
Represent the $\fqm$--points of $\PP^1$ by the couples $\PP^1(\fqm)=\{[1:x] \mid x \in \fqm\} \cup \{P_\infty\}$ for $P_\infty=[0:1]$. Take $\calP=\{[1:x_1],[1:x_2],\dots,[1:x_n]\}$ and $D=(r-1)P_\infty$. Finally, the polynomial $g$ can be seen as a function on $\PP^1$ which lies in $L(rP_\infty)$ but not in $L((r-1)P_\infty)$. Then both constructions match: $\Gamma_r(\mathbf{x},g)=\Gamma(\calP,D,g)$.

\subsubsection{Relation with Cartier codes.} Cartier codes \cite{Cou14} are also defined as a geometric realization of Goppa codes, since well--known properties of Goppa codes naturally extend to them. For instance, Theorem \ref{thm:Wirtz} holds for Cartier code without assumption on the degree of the divisor.

The link with Goppa--like AG codes is the following: by definition, a Cartier code is a subcode of the subfield subcode of a residue code (see \cite[Proposition 4.3]{Cou14}), which actually means that for the good choice of divisor, a Cartier code is a subcode of the corresponding Goppa--like AG code. Moreover, \cite[Theorem 5.1]{Cou14} provides a sufficient condition for both constructions to be equal. More precisely, let us consider a Goppa--like AG code $\Gamma(\calP,D,g)$  and set $G := D+(g)$. Then the Cartier code $\mathrm{Car}_q(\calP,G)$ (see \cite[Definition 4.2]{Cou14}) satisfies $\mathrm{Car}_q(\calP,G) \subseteq \Gamma(\calP,D,g)$, and 
$$ \dim_{\fq} \left( \Gamma(\calP,D,g)/ \mathrm{Car}_q(\calP,G)\right) \leq m \cdot i(G_1),$$
where $G_1$ is any divisor such that 
\begin{equation} \tag{\ref{eq:divisor_G_1}}
G \geq qG_1 \ \mathrm{and} \ G \geq G_1.
\end{equation}
 Above, $i(G_1)$ stands for the index of speciality of $G_1$ (see \cite[Definition~1.6.10]{Sti09}). By the Riemann--Roch theorem, if $\deg(G_1) > 2\mathfrak{g}-2$, then $i(G_1) =0.$ 
Thus, using Remark \ref{rk:G/q}, the Cartier code $\mathrm{Car}_q(\calP,G)$ coincides with the Goppa--like AG code $\Gamma(\calP,D,g)$ whenever $\deg\left[\dfrac{G}{q} \right] > 2\mathfrak{g}-2$. 

\subsection{On the dimension of the square of the dual of a Goppa--like AG code} \label{section:1st_improvement}
In this section, we aim to generalize the properties found by the authors of \cite{MT21} in Section 6, in the context of Goppa--like AG code. Let us consider the AG code 
$$\calC := \calC_{\calL}(\calX,\calP,D+(g))$$
as defined in Section \ref{def:Goppa--like_AG_code}. Applying Equation \eqref{eq:key_equation} yields   
\begin{equation} \label{eq:key_equation_Goppa--like} 
\Tr{\calC}^{\star 2} = (\Gamma(\calP,D,g)^{\perp})^{\star2} \subseteq \sum\limits_{i=0}^{\lfloor{m/2} \rfloor} \Tr{\calC\star \calC^{q^i}}.
\end{equation}

Below, we discuss how to improve the upper bound given in Corollary \ref{coro:1st_bound_mumford}, which is valid for all subfield subcodes of AG codes. The idea is to use the specific algebraic structure of our code inherited from the choice of its divisor.
\noindent In fact, notice that the code $\calC$ is monomially equivalent to $C_{\calL}(\calX,\calP,D)$. 
More precisely, using the equality $\calL(D+(g)) = g^{-1} \cdot \calL(D),$
we can easily deduce the following result.

\begin{lemma} \label{lem:bound_dim_Tr(C*C^q^i)}
  Suppose $s \geq \mathfrak{g}$. Then for all $i \geq 0$, we have 
   $$\dim_{\fq} \Tr{\calC\star \calC^{q^i}} \leq m\left(s\left(q^i+1\right)+1-\mathfrak{g}\right).$$
\end{lemma}
\begin{proof}
From Lemma~\ref{lem:Schur-Product-Power}, we deduce that 
\[\dim_{\fqm} \calC\star \calC^{q^i} \leq \dim_{\fqm} \calL((q^i+1)G) = \dim_{\fqm} \calL((q^i+1)D) = s(q^i+1)+1-\mathfrak{g},\]
the last equality coming from the Riemann--Roch theorem (since $\deg (q^i+1)D =s \geq 2\mathfrak{g}-1$). The result follows from the usual upper bound on the dimension of the trace of a code. 
\end{proof}

\begin{remark}\label{rmk:improvements}
	At first glance, it seems that we could have benefited from Theorem~\ref{thm:dim_ssag} to get a sharper bound in the previous lemma. Indeed, for every $i \geq 1$, we have 
	\[\left[\frac{(q^i+1)G}{q}\right]=q^{i-1}G^+ - (q^i+1)G^-,\]
	writing $G=G^+-G^-$ with $G^+, \: G^- \geq 0$. However, in the context of Goppa--like codes, we have $G=D+(g)$ where $g \notin \calL(D)$, hence $G^- \neq 0$. Without further hypotheses on the divisor $D$ and the function $g$, the degree of the divisor $\left[\frac{(q^i+1)G}{q}\right]$ may be too low to bound the dimension of its Riemman--Roch space from below via the Riemann--Roch theorem.
\end{remark}

This simple lemma yields an upper bound on the dimension of the square of the dual of Goppa--like codes.
\begin{proposition} \label{prop:bound_dim_using_inclusions}
    Let $\calC := \calC_{\calL}(\calX,\calP,D+(g))$ be an AG code as above, and suppose $s \geq \mathfrak{g}$. Set $k := \dim_{\fqm}\calC$ and  $e := \min\left(\left\lfloor \frac{m}{2} \right\rfloor,\left\lfloor \log_q\left(\frac{k^2}{s}\right)\right\rfloor\right)$. Then
    $$\dim_{\fq} (\Gamma(\calP,D,g)^{\perp})^{\star 2} \leq \binom{mk+1}{2} - \dfrac{m}{2}\left(k(k-1)(2e+1)-2s\left(\dfrac{q^{e+1}-1}{q-1}\right)\right).$$
\end{proposition}
\begin{proof}
 For any $e \in \set{0,\dots,\lfloor \frac{m}{2} \rfloor}$, Proposition \ref{prop:Tr_BoundSchurSquare} implies that
    \begin{align*}
        \dim_{\fq}(\Gamma(\calP,D,g)^{\perp})^{\star 2}
        & \leq \sum\limits_{i=0}^{\lfloor m/2 \rfloor} \dim_{\fq} \Tr{\calC \star \calC^{q^i}} &\\
        & \leq \sum\limits_{i=0}^{e} m(s(q^i+1)+1-\mathfrak{g})  + \sum\limits_{i=e+1}^{\lfloor m/2 \rfloor} \Tr{\calC \star \calC^{q^i}} & \text{(by Lemma \ref{lem:bound_dim_Tr(C*C^q^i)})}\\
        & \leq \sum\limits_{i=0}^{e} m(sq^i+k) + \left( \frac{m-1}{2} -e \right)mk^2 & \text{(by the Riemann--Roch theorem)}\\
        & \leq \frac{m}{2}\left(2k(e+1)+2s\left(\dfrac{q^{e+1}-1}{q-1}\right)+k^2(m-1)-2ek^2  \right) &\\
        & \leq \binom{mk+1}{2} -  \dfrac{m}{2}\left(k(k-1)(2e+1)-2s\left(\dfrac{q^{e+1}-1}{q-1}\right)\right).&
    \end{align*}
To get the best bound, we maximize the expression $$ \dfrac{m}{2}\left(k(k-1)(2e+1)-2s\left(\dfrac{q^{e+1}-1}{q-1}\right)\right)$$ with respect to $e$. Removing the constant parts, this is equivalent to find the maximum of the function
$$T(e) = ek^2-s\dfrac{q^{e+1}}{q-1}$$
over $\set{0,\dots,\lfloor \frac{m}{2} \rfloor}$ in the discrete domain of non-negative integers.  
We compute the discrete derivative:
\begin{align*}
    \Delta T(e) = T(e+1)-T(e) &= (e+1)k^2-s\dfrac{q^{e+2}}{q-1} - \left(ek^2-s\dfrac{q^{e+1}}{q-1}\right) \\
                              &= k^2 - sq^{e+1}.
\end{align*}
This function is decreasing with $e$, and the smallest value for which $\Delta T(e) \leq 0$ corresponds to its maximum. It is the smallest value of $e$ such that $k^2 \leq sq^{e+1}$, \emph{i.e.}
$$e =  \left\lfloor \log_q\left(\dfrac{k^2}{s}\right)\right\rfloor.$$
\end{proof}

\subsubsection{Sharpness of the bound} \label{section:sharpness}

Definition \ref{def:Goppa--like_AG_code} of a Goppa--like AG code $\Gamma(\calP,D,g) := \calC^{\perp}|_{\fq}$ requires very few hypotheses. Besides the conditions on the supports of $D$ and $(g)$, which guarantee that the code is well--defined, we only ask for the function $g$ not to belong to the Riemann--Roch space $\calL(D)$. This hypothesis is enforced to make sure that the dimension of $\Tr{\calC}^{\star 2}$ is not abnormally small compared to the expected value given in Corollary \ref{coro:first_bound_square_of_trace}, and thus to make Goppa--like AG codes resistant to a distinguisher based on the square of their dual.

First, if the function $g$ lied in $\calL(D)$ (or more generally if the vector of the evaluations $(g(P))_{P \in \calP}$ belonged to $C_\calL(\calX,\calP,D)$), then the code $\calC= \calC_{\calL}(\calX,\calP,D+(g))$ would contain the evaluation of the constant function $1=\frac{g}{g}$, \textit{i.e.} the unit vector $(1,\dots,1)$. In this case, the vector $(1,\dots,1)$ would belong to $\calC^{q^i}$ for every $i \in \{0,\dots, \lfloor{m/2} \rfloor\}$ and each term in the sum on the right hand--side would contain a copy of $\Tr{\calC}$. This non--trivial intersection between the codes $\Tr{\calC\star \calC^{q^i}}$ would contribute with a negative term in the above bound.
Secondly, if $g$ belonged to $\calL(D)$, then $D+(g)$ would be effective. This would imply the inclusion $\calL((q^i+1)D) \subset \calL((q^{i+1}+1)D)$ for every $i \geq 0$. Therefore, in the proof of Proposition \ref{prop:bound_dim_using_inclusions}, when bounding from above the dimension of the sum by the sum of the dimensions of the trace codes, we would have no chance to get a sharp bound.

Unfortunately, the condition $g \notin \calL(D)$ does not guarantee that the bound given in Proposition \ref{prop:bound_dim_using_inclusions} is reached. In the following proposition, we detail one situation in which we cannot hope for equality.

\begin{proposition}\label{prop:non-eq}
	Using the same notation as above, set $\calC_1=C_{\calL}\left(\calX,\mathcal{P},\left[ \frac{D+(g)}{q} \right]\right)$ (see Equation \eqref{eq:G/q} for the definition of $\left[ \frac{G}{q} \right]$, given $G$). If $\dim \calC_1 \geq 1$, then the bound given in Proposition \ref{prop:bound_dim_using_inclusions} is not reached.
\end{proposition}

\begin{proof}
Any non--zero codeword $\mathbf{c} \in  \calC_1 \subset \calC$ satisfies $\mathbf{c}^q \in \calC$. As $\mathbf{c}$ lies in $\fqm^n$, we have $\mathbf{c}^{q^m}=\mathbf{c} \in \calC^{q^{m-1}}$. Therefore,  we have $\calC_1 \subseteq \calC^{q^i} \cap \calC^{q^{m-1}}$, and for every $i \in \set{1,\dots,\lfloor \frac{m}{2}\rfloor}$, we have $\calC_1^{q^i} \subseteq \calC^{q^i} \cap \calC^{q^{i-1}}$. Then
\[\Tr{\calC\star\calC_1^{q^i}} \subseteq \Tr{\calC\star\calC^{q^i}} \cap \Tr{\calC\star\calC^{q^{i-1}}}.\]
As a result, each pair of consecutive terms in the sum $\sum\limits_{i=0}^{\lfloor m/2 \rfloor} \Tr{\calC \star \calC^{q^i}}$ has non--trivial intersection. However, using \cite[Theorem~2]{T19}, equality with the upper bound only occurs if
\[\bigcap_{0\leq j \leq \lfloor m/2\rfloor} \left( \sum_{\substack{0\leq i \leq \lfloor m/2 \rfloor\\ i\neq j}} \Tr{\calC \star \calC^{q^i}} \right)= \varnothing.\]
\end{proof}

\begin{remark}
	As noted in Remark \ref{rmk:improvements}, when picking the function $g$ at random, the code $\calC_1$ is likely to be reduced to zero.
\end{remark}

As recalled in Section \ref{subsec:def-Goppa}, Goppa--like AG codes coincide with Cartier code as soon as $\deg \left[ \frac{D+(g)}{q} \right] > 2 \mathfrak{g}-2$. In this case, the code $\calC_1$ has dimension at least $\mathfrak{g}$. This means that when the Goppa--like code is also a Cartier code, the dimension of the square of its dual is very unlikely to meet the bound given in Proposition \ref{prop:bound_dim_using_inclusions}.

\begin{remark}
	When considering Goppa codes for the McEliece cryptosystem, it is common to ask for the polynomial $g$ to have only simple roots. In this case, we have $ \left[ \frac{(r-1)P_\infty+(g)}{q} \right] = -P_\infty$ and the code $\calC_1$ defined in Proposition \ref{prop:non-eq} is always zero. The situation above thus never occurs.
\end{remark}

\section{One--point Goppa--like AG code on $C_{a,b}$-curves}\label{sec:AG-C_a,b}
The bound given in section \ref{section:1st_improvement} can be improved by considering more structured codes, \emph{i.e.} one--point Goppa--like AG codes on $C_{a,b}$ curves.
\subsection{Definition}

Below, we define a specific class of Goppa--like AG codes on a $C_{a,b}$ curve, associated to a divisor which is equivalent to the one--point divisor $sP_\infty$. 
\noindent Throughout the rest of the paper, we fix a $C_{a,b}$ curve $\calX_{a,b}$ as defined in Definition \ref{def:C_ab_curves}.
\begin{definition} \label{def:one--point_Goppa--like_AG_codes_on_C_a,b_curves}
Let $s'>s$ be two integers such that there exists a function $g \in \calL(s'P_\infty)$ with $\degab{g}=s'$. Given a set of points  $\calP \subset \calX_{a,b}(\F_{q^m})$ such that $\calP \cap \Supp(g) = \varnothing$, we define the one--point Goppa--like AG code associated to $\calP,s$ and $g$ as 
\[\Gamma(\calP,sP_\infty,g) := \calC_{\calL}(\calX_{a,b},\calP,(sP_\infty+(g))^{\perp})|_{\fq}.\]
\end{definition}
This definition might be restrictive, but it is reasonable as these codes can be encoded quickly thanks to the nice basis of $\calL(sP_\infty)$ (see Equation \eqref{eq:basis_L(sP_inf)}), which is desirable if we aim to build a McEliece cryptosystem based on this family of codes. Moreover, this property will be key in the upcoming sections as it allows a better understanding of the square of the dual of any one--point Goppa--like AG code, under some condition on $s$ and $s'$. \\

\noindent In the next sections, we generalize the result given in \cite{MT21} in the case of classical Goppa codes, by defining a weighted Euclidean division on the ring $\calO_{P_\infty}$ (see Equation \eqref{eq:O_Pinf}), whose elements are seen as bivariate polynomials. 
\subsection{Weighted Euclidean division}
\noindent The following proposition generalizes the classical Euclidean division of univariate polynomials in the case of function in $\calO_{P_\infty}$ with respect to the weighted degree $\deg_{a,b}$.

\begin{proposition}\label{prop:div_grob}
Let $m$ be a positive integer and $g \in \calO_{P_\infty}$. Write $\LT{g}=x^\beta y^\alpha$ with $\alpha < a$.
For any function $f \in \calO_{P_\infty}$, we can write $f=f_1g+f_2$ with 
\[f_2 \in \calR(g):= \Span{x^u y^v \mid u \leq \beta + b-1 \text{ and } v\leq a-1 \text{ not both }  u \geq \beta \text{ and } v \geq \alpha}.\]
Moreover, we have $\degab{f_2} \leq \degab{f}$ and $\dim_{\fqm} \calR(g) = \degab{g}.$ 
\end{proposition}

\begin{proof}
  Since $f \in \calO_{P_\infty}$, we can see $f$ as a bivariate polynomial in $x$ and $y$ (see Equation \ref{eq:O_Pinf}). In the polynomial ring $\F_{q^m}[x,y]$, we perform the division of $f$ by a Gr\"obner basis of the ideal generated by the equation $f_{a,b}$ of the curve $\calX_{a,b}$ and the polynomial $g$ with respect to the monomial order $\prec$ defined in Equation \eqref{eq:def_mon_order}. The fact that $f_2$ lies in $\calR(g)$ and the result on the dimension of $\calR(g)$ both follow from \cite[Proposition 4]{GH00}.
 
 Finally, if we had $\degab{f} < \degab{f_2}$ with $f=f_1 g +f_2$, this would mean that $\LT{f_2}=-\LT{f_1 g}=\lambda x^uy^v$ for some $\lambda \in \F_{q^m}^*$ with both $u \geq \beta$ and $v \geq \alpha$, which is not possible by definition of $\calR\left(g\right)$.
\end{proof}

\begin{lemma} \label{lem:weighted_division}
Take $i \geq 1$ and $s'>s \geq 0$. Let also $g \in \calL\left(s'P_\infty\right)$ and $f \in \calL\left(\left(s'(q^i+1)-1\right)P_\infty\right)$. Then there exists $f' \in \calR\left(g^{q^i-q^{i-1}+1}\right)$ such that $\Tr{\frac{f}{g^{q^i+1}}} = \Tr{\frac{f'}{g^{q^i+1}}} $.
\end{lemma}

\begin{proof}
By Proposition \ref{prop:div_grob}, we can write $f=f_1 g^{q^i-q^{i-1}+1} +f_2$ with
 $f_2 \in \calR\left(g^{q^i-q^{i-1}+1}\right)$ and $\degab{f_2} \leq \degab{f}$. Therefore
 \[\Tr{\frac{f}{g^{q^i+1}}}=\Tr{\frac{f_1 g^{q^i-q^{i-1}+1}}{g^{q^i+1}}} +\Tr{\frac{f_2}{g^{q^i+1}}}= \Tr{\frac{f_1^qg}{g^{q^i+1}}} +\Tr{\frac{f_2}{g^{q^i+1}}}. \]
 
 By definition, the second term has the expected form. Let us examine the first term. If $f_1=0$, we are done. Otherwise, the definition of $f_1$ gives
$\degab{f_1} =\degab{f} - s'(q^i-q^{i-1}+1)$, and
 \begin{align*}
 \degab{f_1^qg}  &= q \degab{f_1} + s'\\
     &= q\degab{f} - s'(q-1)(q^i+1).
 \end{align*} 
Then  $\degab{f_1^qg} < \degab{f}$ if and only if $\degab{f} < s'(q^i+1)$, which holds by definition of $f$. Repeating the division process on $f_1^qg$, as the weighted degree decreases, we can find a function $f' \in \calR\left(g^{q^i-q^{i-1}+1}\right)$ such that $\Tr{\frac{f}{g^{q^i+1}}} = \Tr{\frac{f'}{g^{q^i+1}}} $.
\end{proof}
\begin{definition} \label{def:T_i's}
For any $1 \leq i \leq \lfloor\frac{m}{2}\rfloor$, we define
$$\calT_i(s,g):= \Tr{g^{-(q^i+1)}\cdot \left( \calR\left(g^{q^i-q^{i-1}+1}\right)  \cap \calL(s(q^i+1)P_\infty)\right)}$$
and we set $$\calT_0(s,g) := \Tr{g^{-2} \cdot \calL(2sP_\infty}).$$
\end{definition}

\noindent Let $i \geq 0$ and $f \in \calL(sP_\infty) \cdot \calL(sP_\infty)^{q^i} \subseteq \calL(s(q^i+1)P_\infty)$. Lemma \ref{lem:weighted_division} entails that 
$$\Tr{\dfrac{f}{q^{i+1}}} \in \calT_i(s,g).$$
Thus, for all $i \in \set{0,\dots,\lfloor \frac{m}{2} \rfloor}$, we have \begin{equation} \label{eq:Tr(C*C^q^i)_dans_T_i}
\Tr{\calC \star \calC^{q^i}} \subseteq \calT_i(s,g),
\end{equation}
which can be used to improve the bound given in Proposition \ref{prop:bound_dim_using_inclusions}, provided that we can efficiently compute the dimension of the trace codes $\calT_i(s,g)$. This is studied in the next section.
\subsection{Upper bound in Goppa--like case}
In the proposition below, we study the intersection 
\begin{equation}\label{eq:def_Mi}
M_i(s,g):=R\left(g^{q^{i}-q^{i-1}+1}\right) \cap \calL(s(q^i+1)P_\infty)
\end{equation}
for every $i \in \set{1,\dots,\lfloor m/2 \rfloor}$ in order to better grasp the trace codes $\calT_i(s,g)$'s introduced in Definition \ref{def:T_i's}.
First we set some notation: fix $i \in \set{1,\dots,\lfloor m/2 \rfloor}$ and write $\LT{g}=x^\beta y^\alpha$ with $a\beta + b\alpha=s'$. By reducing modulo the equation $f_{a,b}$ of $\calX_{a,b}$, we can write $g^{q^i-q^{i-1}+1}$ such that its leading term with respect to the monomial order $\prec$ is
\begin{equation}\label{eq:gi}
 \LT{g^{q^i-q^{i-1}+1}}=x^{\beta_i} y^{\alpha_i}
\end{equation}
 where $\alpha_i \in \set{0,\dots,a-1}$ is the remainder of the Euclidean division of $\alpha (q^i-q^{i-1}+1)$ by $a$ and 
 \begin{equation}\label{eq:value_beta_prime}
  \beta_i=\beta(q^i-q^{i-1}+1) + b \, \frac{\alpha(q^i-q^{i-1}+1)-\alpha_i}{a}= \frac{s'(q^i-q^{i-1}+1)-b\alpha_i}{a}.
 \end{equation}
Depending on the weighted degree $s'$ of $g$, we can compute the exact dimension of the vector space $M_i(s,g)$ defined in Equation \eqref{eq:def_Mi}.

\begin{proposition} \label{prop:dim_M_i's}
    \begin{enumerate}
        \item If $s'(q^i-q^{i-1}+1) > s(q^i+1)+a$, then $M_i(s,g) = \calL(s(q^i+1)P_\infty)$;
        
        \item If $s'(q^i-q^{i-1}+1) \leq s(q^i+1)+1-2\mathfrak{g}_{a,b}$, then $M_i(s,g) = \calR(g^{q^{i}-q^{i-1}+1})$;
        
        \item If there exists $v^* \in \set{1,\dots,\alpha_i-1}$ such that
        $$ s(q^i+1)+a-b(a+v^*-\alpha_i) < s'(q^i-q^{i-1}+1) \leq  s(q^i+1)+a-b(a+v^*-1-\alpha_i),$$
        we have 
          $$\dim_{\fqm}(M_i(s,g)) = \sum\limits_{v=v^*}^{a-1} \left\lfloor \dfrac{s(q^i+1)-bv}{a} \right\rfloor + v^*(\beta_i+b) + a-v^*.$$
        
        \item Else, there exists $v^* \in \set{\alpha_i+1,\dots,a}$ such that
        $$ s(q^i+1)+a-b(v^*-\alpha_i) < s'(q^i-q^{i-1}+1) \leq  s(q^i+1)+a-b(v^*-1-\alpha_i),$$
        in which case 
        $$\dim_{\fqm}(M_i(s,g)) = \sum\limits_{v=v^*}^{a-1} \left\lfloor \dfrac{s(q^i+1)-bv}{a} \right\rfloor + v^*\beta_i + \alpha_ib +a-v^*.$$
    \end{enumerate}
\end{proposition}
\begin{proof}

Using the notation above, we can write
\begin{align*}
\calR\left(g^{q^i-q^{i-1}+1}\right) &:= \mathrm{Span}_{\fqm} \left\{x^uy^v \mid u \leq \beta_i+b-1 , v \leq a-1 \ \mathrm{not \ both} \ u \geq \beta_i \ \mathrm{and} \ v \geq \alpha_i\right\} \\
&= \mathrm{Span}    \left( \begin{array}{c}
         1,x,\dots,x^{\beta_i +b-1},   \\
         \cdots \\
         y^{\alpha_i -1},y^{\alpha_i -1}x,\dots,y^{\alpha_i -1}x^{\beta_i +b-1}, \\
          y^{\alpha_i},y^{\alpha_i}x,\dots,y^{\alpha_i}x^{\beta_i-1}, \\
         \cdots \\
         y^{a-1},y^{a-1}x,\dots,y^{a-1}x^{\beta_i-1}
    \end{array}
    \right)
\end{align*}
Next, we define for any $v \in \set{0,\dots,a-1}$:
$$\ell^i_v := \max \set{u \geq 0 \mid x^uy^v \in \calL(s(q^i+1)P_\infty)} = \left\lfloor \dfrac{s(q^i+1)-bv}{a}\right\rfloor,$$
implying
\begin{equation*}
\calL(s(q^i+1)P_\infty) = \mathrm{Span}    \left( \begin{array}{c}
         1,x,\dots,x^{\ell^i_0},   \\
         y,yx,\dots,yx^{\ell^i_1}, \\
         \cdots \\
         y^{a-1},y^{a-1}x,\dots,y^{a-1}x^{\ell^i_{a-1}}
    \end{array}
    \right).
\end{equation*}
We thus have a description of a basis of both spaces $\calR(g^{q^i-q^{i-1}+1})$ and $\calL(s(q^i+1)P_\infty)$, leading to an exact formula for the dimension of their intersection $M_i(s,g)$ for any value of $i$:
\begin{equation} \label{eq:exact_dimension_M_i}
\dim_{\fqm} M_i(s,g) = \sum\limits_{v=0}^{\alpha_i-1} \min(\beta_i+b,\ell_v^i+1) + \sum\limits_{v=\alpha_i}^{a-1} \min(\beta_i,\ell_v^i+1).
\end{equation}
It remains to compute the corresponding minima with respect to $v$:
\begin{itemize}
    \item[(i)] If $0 \leq v \leq \alpha_i-1$, by using \eqref{eq:value_beta_prime}, we get
    \begin{align*}
        \beta_i+b \leq \ell_v^i +1 \iff s'(q^i-q^{i-1}+1) \leq F(v) := s(q^i+1)+a-b(a+v-\alpha_i).
    \end{align*}
    \item[(ii)] Otherwise, $\alpha_i \leq v \leq a-1$ and
    \begin{align*}
        \beta_i \leq \ell_v^i +1 \iff s'(q^i-q^{i-1}+1) \leq G(v) := s(q^i+1)+a-b(v-\alpha_i).
    \end{align*}
\end{itemize}
Note that both $F$ and $G$ are decreasing with $v$, and we easily check that $F(0) = G(a)$. Thus, we have the following sequence of integers
$$F(\alpha_i-1) \leq \dots \leq F(0) = G(a) \leq G(a-1) \leq \dots \leq G(\alpha_i).$$
Depending on the value of $s'$, there is a few cases  to consider:
\begin{itemize}
    \item $s'(q^i-q^{i-1}+1) >G(\alpha_i)$, in which case $M_i(s,g) = \calL(s(q^i+1)P_\infty)$;
    \item $s'(q^i-q^{i-1}+1) \leq F(\alpha_i-1)$, and $M_i(s,g) = \calR(g^{q^i-q^{i-1}+1})$;
    \item There exists $v^* \in \set{1,...,\alpha_i-1}$ such that $F(v^*) < s'(q^i-q^{i-1}+1) \leq F(v^*-1)$;
    \item There exists $v^* \in \set{\alpha_i,...,a}$ such that $G(v^*) < s'(q^i-q^{i-1}+1) \leq G(v^*-1)$.
\end{itemize}
The formulas for the dimension of $M_i(s,g)$ follows from the above computations and \eqref{eq:exact_dimension_M_i}.
\end{proof}

Note that (1) corresponds to the case where $\calT_i(s,g) = \Tr{g^{-(q^i+1)}\cdot \calL(s(q^i+1)P_\infty)}, $ which will produce the same bound as the one given in Proposition \ref{prop:bound_dim_using_inclusions}.
Instead, we focus on (2), since in this case, we can show some inclusion relations between the $\calT_i(s,g)$'s.

\begin{proposition} \label{prop:inclusion_T_i's} Let $i^* \in \set{0,\dots,\lfloor\frac{m}{2}\rfloor-1}$ be the smallest integer such that 
\begin{equation} \label{eq:cond_s_s'_i_star}
sq^{i^*} \geq (s'-s)(q^{i^*+1}-q^{i^{*}}+1)+2\mathfrak{g}_{a,b}-1.
\end{equation} 
Then
$$\calT_{i^*}(s,g) \subseteq \calT_{i^*+1}(s,g) \subseteq \dots \subseteq \calT_{\lfloor \frac{m}{2}\rfloor}(s,g).$$
\end{proposition}
\begin{proof}
From Proposition \ref{prop:dim_M_i's} (2), we know that \eqref{eq:cond_s_s'_i_star} implies $$M_{i^*+1}(s,g) = R(g^{q^{i^*+1}-q^{i^*+1-1}+1}).$$ 
Since the function $$i \mapsto \dfrac{s(q^i+1)+1-2\mathfrak{g}_{a,b}}{q^i-q^{i-1}+1}$$ 
is increasing with $i$, we also have
\begin{equation} \label{eq:structure_M_i}
M_{i}(s,g) = R(g^{q^{i}-q^{i-1}+1}), \ \forall \ i \in \set{i^*,\dots,\left\lfloor \frac{m}{2}\right\rfloor+1}.
\end{equation} 
We now prove the inclusions between the $\calT_i's$, assuming first that $i^* \neq 0$ (since the definition of $\calT_0$ is a bit different). Let $i \in \set{i^*,\dots,\lfloor \frac{m}{2}-1 \rfloor}$, and recall that 
$$\calT_i(s,g) := \set{\Tr{\frac{f}{g^{q^i+1}}} \mid f \in \calR(g^{q^i-q^{i-1}+1})\cap \calL(s(q^i+1)P_\infty)}.$$
Given $\Tr{\frac{f}{g^{q^i+1}}}$ in $\calT_i(s,g)$, we want to show that it belongs to $\calT_{i+1}(s,g)$. Applying Proposition \ref{prop:div_grob} by replacing $f$ with $fg^{q^{i+1}-q^i}$ and $g$ by $g^{q^{i+1}-q^i+1}$, we obtain
\begin{equation} \label{eq:division_i}
fg^{q^{i+1}-q^i} = f_1g^{q^{i+1}-q^i+1} + f_2,
\end{equation}
with $f_2 \in \calR(g^{q^{i+1}-q^i+1}) = M_i(s,g)$ (using \eqref{eq:structure_M_i}) and $\degab{f_2} \leq \degab{fg^{q^{i+1}-q^i}}$. Next, we write
\begin{align*}
    \Tr{\frac{f}{g^{q^i+1}}} &= \Tr{\frac{fg^{q^{i+1}-q^i}}{g^{q^{i+1}+1}}} \\
                             &= \Tr{\frac{f_1g^{q^{i+1}-q^i+1}}{g^{q^{i+1}+1}}} + \Tr{\frac{f_2}{g^{q^{i+1}+1}}} \\
                             &= \Tr{\frac{f_1^qg}{g^{q^{i+1}+1}}} + \Tr{\frac{f_2}{g^{q^{i+1}+1}}}.
\end{align*}
By assumption, we immediately have that $\Tr{\frac{f_2}{g^{q^{i+1}+1}}} \in \calT_{i+1}(s,g).$
If $f_1=0$, we are done. Otherwise, we have from \eqref{eq:division_i}:
$$\degab{f_1} = \degab{fg^{q^{i+1}-q^i}} - \degab{g^{q^{i+1}-q^i+1}} = \degab{f}-s'.$$
Thus
\begin{align*}
     \degab{f_1^qg} < \degab{fg^{q^{i+1}-q^i}} & \iff q\degab{f} +(1-q)s' < \degab{f} +s'(q^{i+1}-q^i) \\
                                               & \iff \degab{f} < s'(q^{i}+1),
\end{align*}
which is true since in particular $f \in \calL(s(q^i+1)P_\infty)$ and $s<s'$. Since the weighted degree decreases, we can repeat the division process until eventually we obtain a quotient $f_1$ equal to zero, which proves that $\calT_i(s,g) \subseteq \calT_{i+1}(s,g)$.

In the case $i^*=0$, we also have to prove that $\calT_0(s,g) \subseteq \calT_1(s,g)$, which differs from the other cases due to the definition of $\calT_0$. Let $\Tr{\frac{f}{g^2}} \in \calT_0(s,g)$, for some $f \in \calL(2sP_\infty)$. Using Proposition \ref{prop:div_grob}, this time replacing $f$ with $fg^{q-1}$ and $g$ with $g^{q+1}$ yields
$$fg^{q-1} = f_1g^q + f_2,$$ with $f_2 \in \calR(g^q) = M_1(s,g)$ (using \eqref{eq:structure_M_i} again).
Thus, we can write
    $$ \Tr{\frac{f}{g^2}} = \Tr{\frac{f_1^qg}{g^{q+1}}}  + \Tr{\frac{f_2}{g^{q+1}}}, $$
with $\Tr{\frac{f_2}{g^{q+1}}} \in \calT_1(s,g)$. Since $\degab{f_1} = \degab{fg^{q-1}} - \degab{g^q} = \degab{f}-s'$, we have 
\begin{align*}
     \degab{f_1^qg} < \degab{fg^{q-1}} & \iff q\degab{f} +(1-q)s' < \degab{f} + s'(q-1)\\
                                               & \iff (q-1)\degab{f} < 2s'(q-1)\\
                                               & \iff \degab{f} < 2s',
\end{align*}
which holds since $s<s'$ and $f \in \calL(2sP_\infty)$. Repeating the division process until we found a quotient equal to zero shows that $\calT_0(s,g) \subseteq \calT_1(s,g)$. The other inclusions hold as in the case $i^* \geq 1$.
\end{proof}
Combining \eqref{eq:Tr(C*C^q^i)_dans_T_i} with both the above propositions lead to a better understanding of the dimension of the square of the dual of one--point Goppa--like AG codes. 
\begin{coro} \label{coro:folklore_upper_bound}
 With the same notation as in Proposition \ref{prop:inclusion_T_i's}, set $k:=\dim_{\fqm}C_{\calL}(\calX_{a,b},\calP,sP_\infty+(g))$.
 Then, for all $e \in \set{0,\dots,\lfloor \frac{m}{2} \rfloor}$, the dimension of $\Gamma(\calP,sP_\infty,g)^{\perp})^{\star 2}$ is bounded from above by
 
\[\dim_{\fq} \Gamma(\calP,sP_\infty,g)^{\perp})^{\star 2} \leq   \left(\frac{m-1}{2}-e\right)mk^2+\dim_{\fq}\left(\sum\limits_{i=0}^e \calT_i(s,g) \right). \]
 Moreover, if $i^* \leq e \leq \lfloor \frac{m}{2} \rfloor$, we have 
 \begin{align*}
  \dim_{\fq} \Gamma(\calP,sP_\infty,g)^{\perp})^{\star 2} \leq   & \left(\frac{m-1}{2}-e\right)mk^2 + ms'(q^e-q^{e-1}+1) \\
  & + \dim_{\fq}\left(\sum\limits_{i=0}^{i^*-1} \calT_i(s,g)\right) - \dim_{\fq} \left( \calT_e(s,g) \cap   \sum\limits_{i=0}^{i^*-1} \calT_i(s,g)\right).
 \end{align*}

\end{coro}
%
\begin{proof}
From \eqref{eq:key_equation} and under the assumption on $s$ and $s'$, we have
        \begin{align*}
        \dim_{\fq} (\Gamma_s(\calP,g)^{\perp})^{\star 2}
        & \leq \sum\limits_{i=0}^{\lfloor m/2 \rfloor} \dim_{\fq} \Tr{\calC \star                     \calC^{q^i}} \\
        & \leq \dim_{\fq} \sum\limits_{i=0}^{e}\calT_i(s,g) + \sum\limits_{i=e+1}^{\lfloor m/2 \rfloor} \Tr{\calC \star \calC^{q^i}} \\
                & \leq \dim_{\fq} \sum\limits_{i=0}^{e}\calT_i(s,g) + \left( \frac{m-1}{2} -e \right)mk^2,
        \end{align*}
 for all $e \in \set{0,\dots,\lfloor \frac{m}{2} \rfloor}$. If $i^* \leq e \leq \lfloor \frac{m}{2} \rfloor$,  Proposition \ref{prop:inclusion_T_i's} gives
\[\sum\limits_{i=i^*}^e \calT_i(s,g) = \calT_e(s,g),\]
which concludes the proof.
\end{proof}

Despite the fact that the upper bound given in Corollary \ref{coro:folklore_upper_bound} can be numerically computed with the knowledge of the degree $s$ and the function $g$, it is hard to give a close formula for any parameters, since the intersections of the trace codes $\calT_i(s,g)$'s are hard to manipulate. However, if we suppose that $i^*=0$, we can sharpen the above result.

\begin{thm} \label{thm:bound_with_T_i's_inclusion} 
Suppose that $s \geq (s'-s)q+2\mathfrak{g}_{a,b}-1$ and let $e^* := \min\left(\left\lfloor \frac{m}{2} \right\rfloor, \left\lceil \log_q\left(\frac{k^2}{s'(q-1)^2}\right)\right\rceil+1\right)$. Then
$$\dim_{\fq} (\Gamma(\calP,sP_\infty,g)^{\perp})^{\star 2}\leq \binom{mk+1}{2} - \dfrac{m}{2}(k^2(2e^*+1)+k-2s'(q^{e^*}-q^{e^*-1}+1)). $$
\end{thm}

\begin{proof}
The condition $s \geq (s'-s)q+2\mathfrak{g}_{a,b}-1$ exactly implies that $i^*=0$ and $$\calT_0(s,g) \subseteq \calT_1(s,g) \subseteq \dots \subseteq \calT_{\lfloor \frac{m}{2}\rfloor}(s,g),$$ 
by Proposition~\ref{prop:inclusion_T_i's}. Thus, using Corollary~\ref{coro:folklore_upper_bound} and the inequality $\dim_{\fq}\calT_e(s,g) \leq m \dim_{\fqm} \calR(g^{q^e-q^{e-1}+1})$, we get
\begin{align*}
        \dim_{\fq} (\Gamma(\calP,sP_\infty,g)^{\perp})^{\star 2}
        & \leq \min \left(ms'(q^e-q^{e-1}+1) + \left( \frac{m-1}{2} -e \right)mk^2 \right)\\
        & \leq \min \left(\frac{m}{2}\left(2s'(q^e-q^{e-1}+1)+k^2(m-1)-2k^2e  \right)\right) \\
        & \leq \min\left(\binom{mk+1}{2} - \dfrac{m}{2}\left(k^2(2e+1)+k-2s'(q^e-q^{e-1}+1)\right)\right).
\end{align*}
the minimum being taking over $e \in \set{1,\dots,\lfloor \frac{m}{2} \rfloor}$. 
To get the best bound, we need to maximize the function
$$T(e) = ek^2-s'(q^e-q^{e-1}+1)$$
over $\set{1,\dots,\lfloor \frac{m}{2} \rfloor}$.
We compute the discrete derivative:
\begin{align*}
    \Delta T(e) = T(e+1)-T(e) &= (e+1)k^2- s'(q^{e+1}-q^e+1) - ek^2 + s'(q^e-q^{e-1}+1) \\
                              &= k^2 - s'q^{e-1}(q-1)^2.
\end{align*}
This function is decreasing with $e$, and the smallest value for which $\Delta T(e) \leq 0$ corresponds to its maximum. It is the smallest value of $e$ such that $k^2 \leq s'q^{e-1}(q-1)^2$, \emph{i.e.}
$$e =  \left\lceil \log_q\left(\dfrac{k^2}{s'(q-1)^2}\right)\right\rceil+1.$$
\end{proof}

Several computational experiments showed then when the code $\calC:=\calC_{\calL}(\calX_{a,b},\calP,sP_\infty+(g))$ is sufficiently random the bound given in Theorem \ref{thm:bound_with_T_i's_inclusion} is sharp, leading to a distinguisher if the parameters of $\calC$ are not well--chosen. We give a concrete example showing the sharpness of our bound.

\begin{expl}
Set $q=3$ and $m = 3$. We consider the curve $\calX$ over $\fqm = \F_{729}$ defined by $$ y^2+y = x^3+x+2.$$
This elliptic curve $\calX$ is a particular case of $C_{2,3}$ curve with genus $\mathfrak{g}=1$. Set $s'=s+1$ for $s \geq 0$, and $g \in \fqm(\calX)$ such that $g=x^{\beta}y^{\alpha} + g'$, where $a\beta+b\alpha=s+1$; and $g'$ is sampled at random in $\calL(sP_\infty)$. For each such $g$, consider $\calP_g := \calX(\fqm) \backslash \Supp(g)$. Using {\scshape{Magma}}, we then compare the true dimension of the square of the dual of $\calC_g := \Gamma(\calP_g,sP_\infty,g)$ with the upper bound given in Theorem \ref{thm:bound_with_T_i's_inclusion} for $s \in \set{4,\dots,10}$. Results can be found in \emph{Table \ref{table:expl_sharpness}}.

\begin{table}[h]
\begin{center}
\begin{tabular}{|c|c|c|c|c|c|}
    \hline
   n &$s$&$\dim_{\fq}\calC_g$ & $\dim_{\fq}(\calC_g)^{\star 2}$&$\dim_{\fq}(\calC_g^{\perp})^{\star2}$ & Upper bound in Theorem \ref{thm:bound_with_T_i's_inclusion}\\
    \hline \hline
    $781$ &$4$& $757$& $781$&$234$ & $234$ \\
    \hline 
    $783$ &$5$& $753$& $783$ &$327$ & $327$   \\
    \hline \hline
    $782$ &$6$& $746$&$782$ &$402$ & $402$  \\
    \hline
    $783$ &$7$& $741$& $783$&$483$ & $483$  \\
    \hline \hline
    $782$ &$8$& $734$& $782$&$570$ & $570$   \\
    \hline
    $782$ &$9$& $728$& $782$&$663$ & $663$ \\
    \hline
    $781$ &$10$& $721$ & $781$&$762$ & $762$ \\
    \hline
\end{tabular}
\vspace*{0.3em}
\caption{Sharpness of the bound.} \label{table:expl_sharpness}
\end{center}
\end{table}

In our computing experiments, we can check that $g$ always has simple zeros, hence $\left[ \frac{sP_\infty+(g)}{3} \right] =-P_\infty$. This example illustrates how the bound can be sharp when we are outside the scope of Proposition \ref{prop:non-eq}. 

\end{expl}

In the last section, we will discuss how to efficiently choose the parameters of a one--point Goppa--like AG code in order to resist this distinguisher.

\section{Analysis of the distinguisher}\label{sec:analysis}

In the previous section, we provided an (experimentally) sharp upper bound on the dimension of the square of the dual of a one--point Goppa--like code, which could lead to a distinguisher for the corresponding code. More precisely, let $\calC := C_\calL(\calX_{a,b},\calP,sP_\infty +(g))$ be an AG code as above, with $\degab{g} = s'>s\geq 2\mathfrak{g}_{a,b}-1$. We showed that if $s$ and $s'$ are such that $s \geq (s'-s)q+2\mathfrak{g}_{a,b}-1$, then
\begin{equation} \label{eq:best_upper_bound}
\dim_{\fq} (\Gamma(\calP,sP_\infty,g)^{\perp})^{\star 2} \leq \min \left(\frac{m}{2}\left(2s'(q^{e^*}-q^{e^*-1}+1)+k^2(m-1-2e^*)  \right),n\right),
\end{equation}
where $e^* := \left\lceil \log_q\left(\dfrac{k^2}{s'(q-1)^2}\right)\right\rceil+1$. Thus, the code is distinguishable from a random one only if the right hand--side of \eqref{eq:best_upper_bound} is smaller than the lenght $n$ of the code. It is possible to study when this case occurs, by starting to bound from above the maximal possible lenght: since $\calP \cap \Supp(g) = \varnothing,$ this maximum is reached when $\calP = \calX_{a,b}(\fqm) \backslash \Supp(g)$ and $g$ has only one zero (\emph{i.e.} $\#\Supp(g)=2$), that is
$$n = \# \calP = |\calX_{a,b}(\fqm)|-2 \leq q^m-1+2\mathfrak{g}_{a,b}\sqrt{q^m},$$
using the Hasse--Weil bound. In order to protect the code against the distinguisher, the parameters have to be chosen such that 
\begin{equation} \label{eq:cond_not_to_distinguish}
\frac{m}{2}\left(2s'(q^{e^*}-q^{e^*-1}+1)+k^2(m-1-2e^*)  \right)\geq q^m-1+2\mathfrak{g}_{a,b}\sqrt{q^m}.
\end{equation}

In what follows, we focus on two specific classes of $C_{a,b}$ curves. First, we determine the maximal (with respect to the dimension) codes we can distinguish in the case where $\calX_{a,b}$ is an elliptic curve. This case is relevant since it is the closest to the case of classical Goppa codes, and we will see that our results are very similar to the one given in \cite{MT21}. Next up, we focus on the particular case of the Hermitian curve, which also turns out to be a $\calX_{a,b}$ curve. It is well--known to be a good candidate to construct efficient codes as it is a maximal curve. Due to its high genus, we then show that any one--point Goppa--like code defined on it cannot be distinguished.

\subsection{High rate distinguishable codes in the case of elliptic curves}

Let $\calX_{a,b}$ be an elliptic curve, \emph{i.e.} $a=2$ and $b=3$. For some set of parameters which produces codes of cryptographic size, we compute the maximal distinguishable value of $s$. To get close to the case of classical Goppa codes, we also fix $s'=s+1$.
\begin{table}[h]
\begin{center}
\begin{tabular}{|c|c|c||c|c|c|c|c|c|}
    \hline
    $q$ & $m$ & $n$ & Largest distinguishable $s$ & Corresponding rate\\
    \hline \hline
     $2$ & $12$ & $4218$ & $14$ & $0,963$ \\
    \hline 
     $2$ & $13$ & $6688$ & $18$ & $0,982$  \\
    \hline \hline
     $3$ & $7$ & $2186$ & $15$ & $0,962$ \\
    \hline
     $3$ & $8$ & $6393$ & $24$ & $0,977$ \\
    \hline \hline
     $5$ & $5$ & $3043$ & $27$ & $0,961$  \\
    \hline
     $5$ & $6$ & $4500$ & $22$ & $0,971$ \\
    \hline
     $5$  & $6$ & $6688$ & $30$ & $0,976$ \\
    \hline \hline
     $7$ & $4$ & $2395$ & $27$ & $0,957$ \\
    \hline
      $7$ & $5$ & $4650$ & $26$ & $0,971$ \\
    \hline
      $7$ & $5$ & $8192$ & $37$ & $0,979$ \\
    \hline \hline
      $17$ & $3$ & $4820$ & $92$ & $0,943$ \\
    \hline
\end{tabular}
\vspace*{0.3em}
\caption{Largest distinguishable Goppa--like AG code in elliptic case.}
\end{center}
\end{table}

As it was noticed in \cite{MT21} and as we can see above, we are only able to distinguish high rate codes. The smallest distinguishable rates are roughtly the same as the one given in \cite{MT21}. 

\subsection{Codes on the Hermitian curve}

As the Hermitian curve is a particular case of $C_{a,b}$ curve, we investigate the behaviour of one--point Goppa--like AG code constructed on it with respect to our distinguisher.
In particular, we show that all these codes resist to it, since  \eqref{eq:cond_not_to_distinguish} always holds in this setting, essentially because the genus of the Hermitian curve is too high with respect to the size of the field. Let us first recall some known results about the Hermitian curve (\cite{Sti09}). 

Let $m \geq 1$  be an even integer and denote by $q_0 := q^{m/2}$, so that $\fqm = \fqo$. The Hermitian curve $\calH$ over $\fqo$ is defined by the equation
$$\calH : y^{q_0}+y = x^{q_0+1}.$$
Its genus is given by $\mathfrak{g}_{\calH} = \frac{q_0(q_0-1)}{2}$ and it is a maximal curve, \emph{i.e.} $\#\calH(\fqo) = q_0^3+1$.

\begin{proposition} \label{prop:Hermitian_Goppa_like_are_secured}
    Suppose $s \geq (s'-s)q+2\mathfrak{g}_{\calH}-1$. Then for any choice of $g$ and $\calP$, the one--point Goppa--like code $\Gamma(\calP,sP_\infty,g)$ resists the distinguisher given in Theorem \ref{thm:bound_with_T_i's_inclusion}.
\end{proposition}

\begin{proof}
    As discussed at the beginning of the section, the code cannot be distinguished whenever \eqref{eq:cond_not_to_distinguish} holds. In this case, we know exactly the number of rational points, and thus the length $n$ of the Goppa--like code is at most $q_0^3-1$. Since $m$ is even, we are left to prove that 
    \begin{equation} \label{eq:dont_distinguish_Hermitian_case}
    \mathfrak{B}(e^*) := ms'(q^{e^*}-q^{e^*-1}+1) + \left( \frac{m}{2}-e^*\right)mk^2 \geq q_0^3-1,
    \end{equation}
    where $k := s+1-\mathfrak{g}_{\calH}$ stands  for the dimension of the corresponding AG code.
    \begin{itemize}
        \item [-] If $e^* < \frac{m}{2}$, then $\mathfrak{B}(e^*) > mk^2$. Using the assumption on $s$ and $s'$, we know that $s \geq 2\mathfrak{g}_{\calH}+q-1$. This yields $k \geq \mathfrak{g}_{\calH}+q$ and thus
        \begin{align*}
\mathfrak{B}(e^*) - (q_0^3-1) 
&> m(\mathfrak{g}_{\calH}^2+2\mathfrak{g}_{\calH}q+q^2)-(q_0^3-1)&\\
& > \frac{m}{4}(q_0^4-2q_0^3+q_0^2) + mq(q_0^2-q_0+q) -q_0^3+1 &\\
& \geq \frac{1}{2}(q_0^4-2q_0^3+q_0^2)+4(q_0^2-q_0+2)-q_0^3+1 \quad &(m\geq 2 \ \mathrm{and} \ q\geq 2)&\\
& > \frac{q_0}{2} (q_0^3-4q_0^2+9q_0-8) > 0,&
        \end{align*}
        since $q_0 \geq 2$. Inequality \eqref{eq:dont_distinguish_Hermitian_case} holds in this case.
        \item[-] If $e^* = \frac{m}{2},$ then since $q_0=q^{m/2}$, we have $\mathfrak{B}\left(\frac{m}{2}\right) = ms'(q_0-q_0q^{-1}+1)$. Moreover, $s'>s$ implies $s' \geq 2\mathfrak{g}_{\calH}+q$, and
        \begin{align*}
           \mathfrak{B}\left(\frac{m}{2}\right) - (q_0^3-1) 
           &\geq m(2\mathfrak{g}_{\calH}+q)(q_0-q_0q^{-1}+1)-q_0^3+1 \\ 
           &\geq 2\left(q_0^2(q_0-1)\left(\frac{q-1}{q}\right)+q_0(q_0-1)+q_0(q-1)+q\right)-q_0^3+1 \\
           & \geq q_0^3\left(2\left(\frac{q-1}{q}\right)-1\right) + 2q_0^2\left(1-\left(\frac{q-1}{q}\right)\right) + 2(q_0(q-2)+q)+1.
        \end{align*}
        Clearly, the last expression is minimal for $q=2$, so we finally gets
        $$\mathfrak{B}\left(\frac{m}{2}\right) \geq q_0^2 + 5 >0,$$
        which proves \eqref{eq:dont_distinguish_Hermitian_case} in this case and conclude the proof.
    \end{itemize}
\end{proof}
Consequently, in the light of Table \ref{table:goppa-herm}, it is still reasonable to consider the Hermitian curve to build efficient SSAG code--based cryptosystem.

\bigskip

\noindent
\textbf{Acknowledgements.} 
This work was funded in part by the grant ANR-21-CE39-0009-BARRACUDA from the French National Research Agency. The first and the third authors are supported by the French government ‘‘Investissements d’Avenir" program ANR-11-LABX-0020-01. Through the LMB, the second author received support from the  EIPHI Graduate School (contract ANR-17-EURE-0002).

\bibliography{biblio}
\bibliographystyle{alpha}

\end{document}